%% file: main.tex
\documentclass[10pt]{article}

\include{preamble}

\title{Grouped approximate control variate estimators}
\author{Alex A. Gorodetsky$^{1*}$ \and John D. Jakeman$^2$ \and Michael S. Eldred$^2$}

\date{
  \small 
    $^1$Department of Aerospace Engineering, University of Michigan, Ann Arbor, MI. 48109\\%
    $^2$Optimization and Uncertainty Quantification, Sandia National Laboratories, Albuquerque, NM, 87123\\
    \vspace{5pt}
    $^*${\it Corresponding author: goroda@umich.edu}\\[2ex]%
  February 2024
}

\input{macros}

\begin{document}
\maketitle

\begin{abstract}
  This paper analyzes the approximate control variate  (ACV) approach to multifidelity uncertainty quantification in the case where weighted estimators are combined to form the components of the ACV. The weighted estimators enable one to precisely group models that share input samples to achieve improved variance reduction. We demonstrate that this viewpoint yields a generalized linear estimator that can assign any weight to any sample. This generalization shows that other linear estimators in the literature, particularly the multilevel best linear unbiased estimator (ML-BLUE) of Schaden and Ullman in 2020, becomes a specific version of the ACV estimator of Gorodetsky, Geraci, Jakeman, and Eldred, 2020. Moreover, this connection enables numerous extensions and insights. For example, we empirically show that having non-independent groups can yield better variance reduction compared to the independent groups used by ML-BLUE. Furthermore, we show that such grouped estimators can use arbitrary weighted estimators, not just the  simple Monte Carlo estimators used in ML-BLUE. Furthermore, the analysis enables the derivation of ML-BLUE directly from a variance reduction perspective, rather than a regression perspective.
\end{abstract}

\section{Introduction}

Multifidelity sampling approaches in uncertainty quantification (UQ) seek to reduce the mean squared error in the estimation of the statistics of a high-fidelity quantity of interest. These methods are essential for tractable model-based analyses that require repeated evaluation of computationally expensive simulations. Consequently, the last decade has seen tremendous growth in multifidelity (MF) approaches for reducing error in UQ statistics. For example, multilevel Monte Carlo (MLMC) ~\cite{giles2008}, multi-index Monte Carlo~\cite{haji2016} (MIMC), multifidelity Monte Carlo (MFMC)~\cite{peherstorfer2016}, and multilevel-multifidelity Monte Carlo~\cite{geraci2017multifidelity} have all emerged as approaches to combine evaluations from ensembles of model instances, where these instances are typically defined by varying model forms and/or discretization levels in order to trade accuracy versus cost.

Many of these existing MF sampling methods were unified by the approximate control variate (ACV) framework~\cite{gorodetsky2020}. This framework considers how linear combinations of statistical estimators can be used to achieve improved performance. These underlying estimators can be of varying types including simple Monte Carlo, importance sampling~\cite{pham2022}, or surrogate based~\cite{yang2023control}. They are also typically instantiated with specific sample allocation strategies across models to ensure that certain estimators are grouped together with respect to the input samples.
Overall, this approach outlined a general framework for optimal sample allocation with fused estimators, with flexibility that can be exploited through additional configuration optimizations such as
directed acyclic graph and ensemble selection schemes~\cite{Bomarito2022}.

Recently, the Multilevel Best Linear Unbiased Estimator (ML-BLUE) approach~\cite{schaden2020,schaden2021} has emerged as a group-based sampling alternative to ACV.
ML-BLUE combines Monte Carlo (MC) estimators from different groupings of models based on linear regression theory, focusing on the specific sampling context of a set of independent MC estimators of groups of models.  In this context, a subset of ACV approaches (e.g., ACV-IS with independent samples) is an ML BLUE.  In a broader context (e.g., overlapping non-independent sample sets), another subset of ACV (e.g. ACV-MF) are outside the class of BLUE.
In this paper, we explore the broader context of the ACV framework --- especially the consideration of sample-grouped estimators --- 
and show that any ML-BLUE estimator is in fact an approximate control variate, whereas not every ACV is a ML-BLUE. While ML-BLUE is best for a specific type of model groupings, it may not be best in the sense of providing the greatest variance reduction amongst all possible groupings. In this paper, we show how ML-BLUEs can be reformulated as ACVs and how this insight enables us to show that all estimators that linearly weight samples from model ensembles are representable as a linear ACV.

To be more specific, we consider the case of unbiased estimators so that all of these approaches consider the following problem:

 \begin{problem}[Multifidelity Variance Reduction]\label{prb} 
   Let $Q$ denote a ``high-fidelity'' random variable, and let $\hat{Q}(c)$ denote an {\it unbiased} estimate of some statistic of $Q$ achieved with cost $c$. Let $L$ ``low-fidelity'' random variables be denoted as $Q_{1},\ldots,Q_{L}$. Estimates of statistics $\hat{Q}_i(c_i)$ of each low fidelity model can also be obtained at a cost $c_i.$  Multifidelity variance reduction strategies seek to find a new estimator $\hat{Q}^*$ for the high-fidelity statistics by fusing estimators of all models subject to a total cost.\footnote{One can also seek the smallest variance for a fixed cost.}
 \end{problem}

 The solution to this problem typically has two components: first one must determine an {\it estimator ansatz}, and, second, one must determine how the costs are allocated (typically through numbers of samples). These tasks are interrelated as the best sampling strategy may depend on the estimator ansatz, and vice versa.
Moreover, an ``optimal'' solution remains open. This paper makes a significant stride in this direction by developing an estimator that can represent {\it any linear weighted combination of samples}. The optimal choice of groupings and sample allocations remains open. In particular, the novel contribution of this paper is the development and analysis of the grouped ACV (GACV) estimator (Definition~\ref{def:gacv}), which provides a more general way to look at the ACV where each of the low-fidelity model estimators estimators can be formed through linear combinations of weighted estimators. Moreover, our development of GACV estimators enabled the following additional contributions, in which we:
\begin{enumerate}
\item Demonstrate that any ML-BLUE defined in~\cite{schaden2020,schaden2021,destouches2023multivariate} and the Ensemble ACV~\cite{pham2022} can be interpreted as specific cases of the GACV (Theorem~\ref{thm:mlblue_is_acv} and ~\eqref{eq:ensemble_acv}, respectively);
\item Generalize the idea of independent model groupings to non-independent groupings and to general statistical estimators (beyond Monte Carlo estimators of mean and variance) in Theorem~\ref{thm:opt_gacv}.
\item Show that the ML-BLUE approach can be derived through variance minimization arguments (Corollary~\ref{corr:blue_variance}).
\item Provide empirical validation that non-independent groupings can yield improved variance reduction (Section~\ref{sec:general}).
\end{enumerate}
Note that while analysis of ML-BLUEs as the solution of a variance reduction problem has previously been performed in the case of mean and variance estimation~\cite{destouches2023multivariate}, our results are focused on the general case of arbitrary estimators. 

The remainder of this paper is structured as follows. In Section~\ref{sec:mlblue_is_acv}, we describe the connection between ACVs and ML-BLUEs that serves to inspire and motivate the general approach described in Section~\ref{sec:generalized_estimators}. The generalized estimator is then analyzed and its optimal {selection of weights} is defined in Section~\ref{sec:analysis}. Finally, Section~\ref{sec:benefits} provides a setting for comparisons between the GACV and the ML-BLUE, and it is shown that the additional generality of the GACV can produce estimators with smaller variance than ML-BLUE estimators in certain cases.

\subsection{Notation}\label{sec:notation}

Sets are denoted by caligraphic upper case, such as $\set{W}.$ Exception is made for ordered sets of real numbers corresponding to vectors and matrices.  Vectors are denoted by bold lowercase, {\it e.g.,} $\vec{w}$, and matrices are denoted by bold uppercase, {\it e.g.,} $\mat{W}$. We make specific exception for zero vectors and matrices where $\mat{0}_{L}$ and $\mat{0}_{LM}$ denote a vector of size $L$ and a matrix of size $L \times M$, respectively. A finite set of integers is denoted by $[L] = \{0,1,\ldots,L\}.$ The power set $2^{\mathcal{S}}$ is the set of all possible subsets of $\mathcal{S}$. 

 The elements of a countable set can be mapped to the integers via an index to provide an ordering. For example,
\begin{equation}
  \set{W} = \{ \mathbf{w}^1,\ldots, \mathbf{w}^k \}, \qquad k \in \mathbb{Z}^+
  \label{eq:set_notation}
\end{equation}
denotes a set of $k$ vectors. %% \mse{(consider $\mathbb{Z}^+$ or allow empty sets)}
. While superscripts are used to define an ordering in a general set, subscripts are used to denote elements of a vector or matrix --- as in $\mathbf{w}_{i}$ or $\mat{W}_{ij}$. In contrast to elements of a vector, the elements of a set may be different objects. For example, the set in~\eqref{eq:set_notation} differs from a matrix because each of the vectors may be of different sizes. 

We will make special use of multi-indices to denote subsets of $[L].$ Consider a subset of $[L]$ with $n$ elements denoted by $\mathcal{S} \in  2^{[L]} \, \backslash \, \emptyset,$ with $|\mathcal{S}| = n.$ A lexigraphical ordering of this set can be associated with a multi-index vector $\vec{\lambda} \in [L]^n$ with elements $\vec{\lambda}_i \in \mathcal{S}$ for $i = 1,\ldots,n,$ and $\vec{\lambda}_{i} < \vec{\lambda}_j$ if $i \leq j$. Given this ordering, an inverse mapping can be established $\vec{\lambda}^{-1}: \mathcal{S} \to [n]$ from an element to its location in the multi-index. Finally, associated with this subset is a  {\it restriction} matrix $\mat{R}$ that determines $\vec{\lambda}$ from the vector $[L]$. That is $\mat{R} \in \reals^{n \times (L+1)}$ with elements
  \begin{equation}
    \mat{R}_{ij} =
    \left\{
    \begin{array}{cc}
      1 & \text{ if } \vec{\lambda}_i  = j \\
      0 & \text{ otherwise} 
    \end{array}
    \right.
    ,
  \end{equation}
  so that $\vec{\lambda} = \mat{R} [L].$ Lastly, $\mathbf{1}(\vec{w}\geq 0)$ denotes a masking  function which outputs a vector with entries that assume the value of one when $\vec{w}_i \geq 0$ and zero otherwise.

\subsection{Linear Approximate Control Variates}

In this section we describe the linear approximate control variate. We begin with a description of the classical control variate from which it is derived. Control variates combine an ensemble of {\it estimators} to obtain a new estimator with reduced variance compared to some baseline. Let the estimators in this ensemble be $\hat{\set{Q}} = \{\hat{Q}$, $\hat{Q}_1, \ldots, \hat{Q}_{L}\},$ and let $\hat{Q}$ denote the baseline estimator whose variance we seek to reduce.

The linear control variate creates a new estimator through the linear combination 
\begin{equation}
\hat{Q}^{CV}(\alpha, \alpha_1, \alpha_{2},\ldots, \alpha_{L},\beta_1,\ldots,\beta_{L}) = \alpha \hat{Q} + \sum_{\ell=1}^{L} \alpha_{\ell} \hat{Q}_{\ell} + \sum_{\ell=1}^{L} \beta_\ell \mu_\ell
\end{equation}
where $\mu_{\ell}$, for $\ell=1,\ldots,L,$ are presumed to be the {\it known} expectations of $Q_1,\ldots,Q_{L}$; and $\alpha,\alpha_{1},\ldots,\alpha_{L},\beta_{1},\ldots,\beta_{L}$ are real-valued tunable coefficients.

The control variate estimator $\hat{Q}^{CV}$ can be made {\it unbiased} (i.e. $\mathbb{E}\left[\hat{Q}^{CV}\right] = \mathbb{E}\left[\hat{Q}\right]$) through two constraints on the coefficients. The first is that the baseline estimator $\hat{Q}$ is unbiased so that we can set $\alpha = 1$. The second is that the $L$ low-fidelity estimators $\hat{Q}_i$ are unbiased so that setting  $\beta_\ell = -\alpha_{\ell}$ cancels out the sum in expectation. Under these constraints the standard CV representation becomes
\begin{equation}
\hat{Q}^{CV}(\alpha_1, \alpha_{2},\ldots, \alpha_{L}) = \hat{Q} + \sum_{\ell=1}^{L} \alpha_{\ell}\left( \hat{Q}_{\ell} - \mu_\ell\right)
\end{equation}
%% where $\mu_{i}$, for $i=1,\ldots,L$ are presumed to be the {\it known} expectations of $Q_1,\ldots,Q_{L}$; $\alpha,\alpha_{1},\ldots,\alpha_{L},\beta_{1},\ldots,\beta_{L}$ are tunable coefficients.
The remaining free parameters $\alpha_{1},\ldots,\alpha_{L}$ can be selected according to the algorithmic strategy, and typically they are chosen to minimize the overall estimator variance.  We refer the reader to~\cite{lavenberg1981perspective,lavenberg1982statistical,rubinstein1985efficiency} for additional details.

The linear control variate requires knowledge of the means $\mu_{1},\ldots,\mu_{L}$. The linear {\it approximate} control variate was introduced to circumvent this issue~\cite{gorodetsky2020}. The approach assumed a %an extremely
simple heuristic: replace $\mu_\ell$, termed the expected control variate mean (ECVM), by another estimator $\hat{Q}_{\ell}^\mu$. To distinguish it from the first term, we refer to the first term as the correlated mean estimator (CME) and denote it by $\hat{Q}_{\ell}^e$ 
so that the approximate control variate is written as~\cite[Eq 7]{gorodetsky2020}
\begin{align}
  \hat{Q}^{ACV}(\alpha_1, \alpha_{2},\ldots, \alpha_{L}) &= \hat{Q} + \sum_{\ell=1}^{L} \alpha_{\ell}\left( \hat{Q}^e_{\ell} - \hat{Q}^\mu_{\ell}\right) = \hat{Q} + \sum_{\ell=1}^{L} \alpha_{\ell} \Delta_{i},
  \label{eq:ACV}
\end{align}
where $\Delta_{\ell} = \hat{Q}^e_{\ell} - \hat{Q}^\mu_{\ell}.$ For $\hat{Q}^{ACV}$ to remain unbiased, it is sufficient to require that either both $\hat{Q}_{\ell}^e$ and $\hat{Q}_{\ell}^\mu$ are unbiased or have the same bias.

Finally, it will often be convenient to treat the baseline estimator in the same framework as the $L$ other alternatives by rewriting~\eqref{eq:ACV} as
\begin{align}
  \hat{Q}^{ACV}(\alpha_1, \alpha_{2},\ldots, \alpha_{L}) &= \sum_{\ell=0}^{L} \alpha_{\ell}\left( \hat{Q}^e_{\ell} - \hat{Q}^\mu_{\ell}\right) = \sum_{\ell=0}^{L} \alpha_{\ell} \Delta_{\ell},
  \label{eq:ACV2}
\end{align}
where $\alpha_{0} = 1$, $Q_0 = Q$, $\hat{Q}^{e}_{0} = \hat{Q} = \Delta_{0}$. Note that in this case we choose $\hat{Q}^\mu_{0} = 0$ for notational convenience, since there is only a single estimator for the high-fidelity model.
 We also note that we have not specified any particular form on $\hat{Q}_{\ell}^{e,\mu}$ -- much of the general theory is agnostic to specific choices, as long as they are unbiased. Again the optimal weights can be chosen to minimize the variance of this estimator, and the reader is referred to~\cite{gorodetsky2020}. Furthermore, we note that the underlying estimators can also be chosen by the user, and need not be restricted to Monte Carlo sampling.

\subsection{Multilevel BLUE estimators}

The ML-BLUE estimator is formed by a linear combination of correlated Monte Carlo estimators. It is defined via $K$ model groups, where each group consists of a set of Monte Carlo estimators for $L+1$ models and different groups consist of estimators formed by independent samples. Specifically the $k$-th group is a couple ($\mathcal{S}^k,\mathcal{Z}^k)$, where the models in the group are a non-empty element of the powerset $\set{S}^k \in 2^{[L]} \, \backslash \, \emptyset $ of size $|\set{S}^k| = n_k.$ Associated with this group is the multi-index vector $\vec{\lambda}^k \in [L]$ and restriction matrix $\mat{R}^{k}$ as discussed in Section~\ref{sec:notation}. The samples used for the Monte Carlo estimators in this group are $\mathcal{Z}^k = \{z_1^k, \ldots, z_{m_k}^k\}.$
Within each group, the ML-BLUE estimator constructs the  Monte Carlo estimators $\hat{Q}_{\ell}^{k} = \frac{1}{m_k} \sum_{i=1}^{m_k} Q_{\ell}(z_{i}^k)$  that are correlated due to shared samples within the group. Note that $\mathcal{Z}^k$ and $\mathcal{Z}^j$ are independent sets so that different groups are not correlated. Finally, the ML-BLUE estimator weights and combines these estimators. The weights for each group are given by $\set{B} = \{\vec{\beta}^1, \ldots, \vec{\beta}^K\}$ with $\vec{\beta}^k \in \reals^{n_k}$, and then the ML-BLUE estimator is given as~\cite[Eq 2.6]{schaden2021}
\begin{equation}
  \hat{Q}^{ML-BLUE}(\set{B}) ~=~
  \sum_{k=1}^K \sum_{\ell \in \set{S}^k} \vec{\beta}_{\vec{\lambda}^{k^{-1}}(\ell)}^{k} \hat{Q}^k_{\ell}
  ~=~
  \sum_{k=1}^K \sum_{\ell \in \set{S}^k} \vec{\beta}_{\vec{\lambda}^{k^{-1}}(\ell)}^{k} \frac{1}{m_k} \sum_{i=1}^{m_k} Q_{\ell}(z^k_{i}).
  \label{eq:blue}
\end{equation}
Note that this estimator is only ``best'' under an optimal choice of weights, and the form above with optimal weights corresponds to the solution of a corresponding least squares regression problem~\cite[Eq. 2.8]{schaden2020}. However, the theory that follows is valid {\it any} choice of weights, including the optimal weights chosen by the least squares problem cited above. Thus, we are slightly generalizing the definition of an ML-BLUE to any weighted set of independent estimators --- it should be understood that many of the desirable properties of the ML-BLUE are in fact due to choosing optimal weights.

It will later be useful to extend each of the coefficient vectors $\vec{\beta}^k$ into a vector of size $L+1$ with zeros filling the elements that correspond to models not included in group $S^k$. These vectors, $\tilde{\vec{\beta}}^k\in \reals^{L+1}$ are obtained through multiplication by the transpose of the corresponding restriction matrices
\begin{equation}
  \tilde{\vec{\beta}}^k = \mat{R}^{k^T} \vec{\beta}^k.
  \label{eq:beta_tilde}
\end{equation}
Using these zero-filled vectors, the ML-BLUE is equivalently
\begin{equation}
  \hat{Q}^{ML-BLUE}(\set{B}) ~=~ \sum_{k=1}^K \sum_{\ell \in \set{S}^k}^{L} \tilde{\vec{\beta}}_{\ell}^{k} \hat{Q}^k_{\ell} ~=~
  \sum_{k=1}^K \sum_{\ell=0}^{L} \tilde{\vec{\beta}}_{\ell}^{k} \hat{Q}^k_{\ell},
  \label{eq:blue2}
\end{equation}
with the convention that $\hat{Q}^k_{\ell} = 0$ if $\ell \notin \set{S}^k$  in the second equality.

An appropriate choice of coefficients $\vec{\beta}_{\ell}^k$ leads to an estimator that both is unbiased and has reduced variance. We will return to the optimal set of coefficients in Section~\ref{sec:analysis}. For now, we will demonstrate that, for any choice of coefficients that yields an {\it unbiased} estimator, the MLBLUE is also an ACV with a particular choice of ensemble estimators.

We begin with a requirement needed for unbiasedness.
\begin{proposition}%% [ML-BLUE Unbiased Requirements]
  Let $\beta_{\ell} = \sum_{k=1}^{K} \vec{\tilde{\beta}}_{\ell}^k$ for $\ell =0,\ldots,L$ be the sum of the weights of each model $\ell$ across the groups. An unbiased estimator of the form~\eqref{eq:blue} requires
  \begin{equation}
    \left(1 - \beta_0\right)\mathbb{E}\left[Q_0\right] = \sum_{\ell=1}^L \mathbb{E}\left[Q_{\ell}\right] \beta_{\ell}
    \label{eq:prop1_result}
  \end{equation}
\end{proposition}
  
\begin{proof}
  To be unbiased, we require
    $\mathbb{E}\left[\hat{Q}^{ML-BLUE}(\beta)\right] = \mathbb{E}[Q_0] = \mathbb{E}[Q].$
  From definition~\eqref{eq:blue2}, we have
  \begin{equation}
    \mathbb{E}\left[\hat{Q}^{ML-BLUE}(\vec{\beta})\right]
    = \sum_{k=1}^K \sum_{\ell = 0}^L \tilde{\vec{\beta}}_{\ell}^{k} \mathbb{E}\left[Q_{\ell}\right].
  \end{equation}
  We now separate the highest fidelity model from the lower fidelity models and re-arrange the sum to obtain
  \begin{align}
    \mathbb{E}\left[\hat{Q}^{MBLUE}(\beta)\right]
    &= \mathbb{E}\left[Q_{0}\right] \sum_{k=1}^K \vec{\beta}_{0}^k + \sum_{\ell =1}^L \mathbb{E}\left[Q_{\ell}\right]  \sum_{k=1}^K  \vec{\beta}_{\ell}^{k}   
    = \mathbb{E}\left[Q_{0}\right]\beta_0 + \sum_{\ell=1}^{L} \mathbb{E}\left[Q_{\ell}\right] \beta_{\ell}.
  \end{align}
 Thus to achieve an unbiased estimate we require
  \begin{equation}
    \mathbb{E}\left[Q_{0}\right] = \mathbb{E}\left[Q_{0}\right]\beta_0 + \sum_{\ell=1}^{L} \mathbb{E}\left[Q_{\ell}\right] \beta_{\ell}.
  \end{equation}
  Equality necessitates that the coefficients satisfy~\eqref{eq:prop1_result},
completing the proof.           
\end{proof}

These conditions on the coefficients needed to satisfy the unbiasedness requirement are clearly dependent on the actual expected values of the quantities of interest. This makes the {\it model-specific} solution uncomputable in practice because these quantities are not available. 
Thus we must ask whether there exists a set of coefficients for which this condition is satisfied {\it regardless of the choice of models}. The linear relationship between coefficients and the expectations makes the following condition obvious.

\begin{proposition}[ML-BLUE Unbiasedness] \label{prop:unbiasedness}
  The ML-BLUE is unbiased if $\beta_{0} = 1$ and $\beta_{\ell} = 0$ for $\ell = 1,\ldots, L.$
\end{proposition}
Moreover, all reported optimal coefficients satisfy these conditions, and are guaranteed to satisfy it due to due to Gauss-Markov-Aitken theorem arguments~\cite{schaden2020,schaden2021}.

\section{ML-BLUE estimators are approximate control variates}~\label{sec:mlblue_is_acv}

When originally proposed, ML-BLUE estimators were not recognized as ACV estimators. However, any ML-BLUE estimator is in fact an approximate control variate. To prove this, we transform the expression of an ML-BLUE estimator into the corresponding estimators $\hat{Q}_{\ell}^e$, $\hat{Q}_{\ell}^\mu$ and weights $\alpha_{\ell}$ that comprise an ACV estimator taking the form of~\eqref{eq:ACV}.
\begin{theorem}[ML-BLUE is an ACV]\label{thm:mlblue_is_acv}
  For any set of coefficients $\set{B}$, groupings $\{\mathcal{S}^1,\ldots,\mathcal{S}^k\}$ and samples $\{\mathcal{Z}^1,\ldots,\mathcal{Z}^K\}$ for which the ML-BLUE estimator~(Eqs.~\eqref{eq:blue} or~\eqref{eq:blue2}) is unbiased, there corresponds 
an equivalent ACV estimator~\eqref{eq:ACV} defined by weights
  \begin{equation}
    \alpha_{\ell} = \sum_{k=1}^K  %\jdj{
    \mathbf{1}(\tilde{\vec{\beta}}_{\ell}^{k} \geq 0)\tilde{\vec{\beta}}_{\ell}^k
    %}
      = -\sum_{k=1}^K  \mathbf{1}(\tilde{\vec{\beta}}_{\ell}^{k} < 0)\tilde{\vec{\beta}}_{\ell}^{k},
    \label{eq:equiv_alpha}
  \end{equation}
  for $\ell = 1,\ldots,L,$ and a set of weighted estimators

    \begin{align}
  \hat{Q} = \sum_{k=1}^K \tilde{\vec{\beta}}_0^k \hat{Q}_0^{k} 
  \qquad
  \hat{Q}_{\ell}^e = \sum_{k=1}^K  \omega_{\ell}^{k,e} \hat{Q}_{\ell}^k 
  \qquad
  \hat{Q}_{\ell}^\mu = \sum_{k=1}^K  \omega_{\ell}^{k,\mu} \hat{Q}_{\ell}^k,
  \label{eq:Qs}
    \end{align}
    where $\hat{Q}_{\ell}^k$ is a Monte Carlo estimator using the samples $\mathcal{Z}^k$.
    The weights are given by $\omega_{\ell}^{k,e} = \mathbf{1}(\tilde{\vec{\beta}}_{\ell}^{k} \geq 0)\tilde{\vec{\beta}}_{\ell}^k / \alpha_{\ell}$ and $\omega_{\ell}^{k,\mu} = -\mathbf{1}(\tilde{\vec{\beta}}_{\ell}^{k} < 0) \tilde{\vec{\beta}}_{\ell}^k / \alpha_{\ell},$  for $\ell=1,\ldots,L$.

\end{theorem}

\begin{proof}
  The proof follows from a straightforward grouping of estimators into those with positive and negative weights. We first rewrite Equation~\eqref{eq:blue2} by reversing the order of the sums, pulling out the high-fidelity model, and separating the positive and negative coefficients
\begin{align}
  \hat{Q}^{ML-BLUE}(\set{B}) &=  \sum_{k=1}^K \tilde{\vec{\beta}}_{0}^{k} \hat{Q}_{0}
  + \sum_{\ell=1}^L \left( \sum_{k=1}^K  \tilde{\vec{\beta}}_{\ell}^{k} \hat{Q}^k_{\ell} \right) \\
&= \sum_{k=1}^K \tilde{\vec{\beta}}_{0}^{k} \hat{Q}_{0}
  + \sum_{\ell =1}^L \left( \sum_{k=1}^K  \mathbf{1}(\tilde{\vec{\beta}}_{\ell}^{k} \geq 0) \tilde{\vec{\beta}}_{\ell}^{k} \hat{Q}^k_{\ell} +  \mathbf{1}(\tilde{\vec{\beta}}_{\ell}^{k} < 0) \tilde{\vec{\beta}}_{\ell}^{k} \hat{Q}^k_{\ell} \right) \\
&=  \sum_{k=1}^K \tilde{\vec{\beta}}_{0}^{k} \frac{1}{m_k} \sum_{i=1}^{m_k} Q_{0}(z^k_{i}) + \sum_{\ell=1}^{L}\alpha_{\ell} \left(\sum_{k=1}^K  \omega_{\ell}^{k,e} \frac{1}{m_k} \sum_{i=1}^{m_k} Q_{\ell}(z^k_{i}) -  \sum_{k=1}^K  \omega_{\ell}^{k,\mu} \frac{1}{m_k} \sum_{i=1}^{m_k} Q_{\ell}(z^k_{i})  \right)  
\end{align}
where the final equality used~\eqref{eq:equiv_alpha} and the weight definitions. Substituting the definitions~\eqref{eq:Qs} we obtain
\begin{align}
  \hat{Q}^{ML-BLUE}(\set{B}) = \hat{Q} + \sum_{\ell=1}^L \alpha_{\ell} \left(\hat{Q}_{\ell}^e - \hat{Q}_{\ell}^{\mu}\right).
\end{align}
The right hand side is equivalent to an ACV estimator, if each of the components can be shown to be unbiased. Note that the weights clearly satisfy the following properties $\sum_{k=1}^K \omega_{\ell}^{k, 1} = 1$, $\omega_{\ell}^{k, 0} > 0$, $\sum_{k=1}^K \omega_{\ell}^{k, 2} = 1$ and $\omega_{\ell}^{k, 2} > 0$. Consequently,
\begin{equation*}
  \mathbb{E}\left[\hat{Q}\right] = \sum_{k=1}^K \tilde{\vec{\beta}}_{0}^{k} \mathbb{E}\left[Q_0\right] = \mathbb{E}\left[Q_0\right] \qquad
  \mathbb{E}\left[\hat{Q}_{\ell}^e\right] = \sum_{k=1}^K  \omega_{\ell}^{k,e}\mathbb{E}\left[Q_{\ell}\right] = \mathbb{E}\left[Q_{\ell}\right]  \qquad
  \mathbb{E}\left[\hat{Q}_{\ell}^\mu\right] = \sum_{k=1}^K  \omega_{\ell}^{k,\mu}\mathbb{E}\left[Q_{\ell}\right] = \mathbb{E}\left[Q_{\ell}\right].
\end{equation*}
The last equality of the first expression for $\mathbb{E}[\hat{Q}]$ arises due to Proposition~\ref{prop:unbiasedness}. The other expressions follow from the summation properties of the weights.
\end{proof}

\section{Generalized estimators} \label{sec:generalized_estimators}
The connection between ML-BLUEs and ACVs suggests two generalizations: (1) we can use general estimators within a grouping strategy, and (2) group estimations may benefit from relaxing the independence requirement and reusing samples across correlated input sets. Note that this last benefit has already been observed within the class of ACV estimators, e.g., the ACV-MF estimator~\cite[Def. 3]{gorodetsky2020} can outperform the ACV-IS estimator~\cite[Sec 5.4]{schaden2020}
    when exchanging sample independence for sample reuse%% }
      \footnote{Note that the ACV-IS estimator described in~\cite[Sec. 5.4]{schaden2020} is slightly different from ACV-IS in~\cite[Def. 2]{gorodetsky2020} because the ACV paper reuses a small portion of samples for the ECVM. The authors believe that the version in~\cite{schaden2020} is more in the spirit of the ``independent'' aspect of the ACV-IS, and we follow their work here. However, empirical results in both papers suggest similar behavior in both cases.}.

Using the idea of groupings and the connection we established between ML-BLUE and ACV, we now provide several novel estimators based on what we call the grouped ACV ansatz
\begin{align}
  \hat{Q}^{GACV} &= \sum_{k=1}^K \omega_{0}^{k,e} \hat{Q}_0^k + \sum_{\ell=1}^L \alpha_{\ell}\left(
  \sum_{k=1}^{K} \omega_{\ell}^{k,e}\hat{Q}_{\ell}^{k,e} - \sum_{k=1}^{K} \omega_{\ell}^{k,\mu}\hat{Q}_{\ell}^{k,\mu}\right).
  \label{eq:gacv1}
\end{align}
Like ML-BLUE, there are $K$ groups, but now there is no requirement for each set of group samples to be independent. Nor is there any requirement for the individual estimators to arise from Monte Carlo.
Moreover, both the weights on the $e$ and $\mu$ terms may be nonzero for the same $\ell$ and $k$ if $\hat{Q}_{\ell}^{k,e}$ and $\hat{Q}_{\ell}^{k,\mu}$ do not use identical samples.% 

To make \eqref{eq:gacv1} easier to work with, first note that \eqref{eq:gacv1} is an unbiased estimator if $\hat{Q}_0^k$ and $\hat{Q}_{\ell}^{k,(e,\mu)}$  are unbiased (or the latter have the same bias) for $\ell=1,\ldots,L$ while $\sum_{k=1}^{K} \omega_{\ell}^{k,e} = \sum_{k=1}^{K} \omega_{\ell}^{k,\mu}$ and $\sum_{k=1}^K \omega_{0}^{k,\mu} = 1.$ Second, note that the control variate weight $\alpha_{\ell}$ can be embedded into the coefficients $\omega$, with no change in constraints, so that we can write
\begin{align}
  \hat{Q}^{GACV} &= \sum_{k=1}^K \omega_{0}^{k,e} \hat{Q}_0^k + \sum_{\ell=1}^L \left(
  \sum_{k=1}^{K} \omega_{\ell}^{k,e}\hat{Q}_{\ell}^{k,e} - \sum_{k=1}^{K} \omega_{\ell}^{k,\mu}\hat{Q}_{\ell}^{k,\mu}\right),
\end{align}
This expression can be generalized further by noting that we implicitly assumed $\hat{Q}_{\ell}^{k,e} \neq \hat{Q}_{\ell}^{k,\mu}$, since if they were equal, these two terms would be combined. Consequently, the second term can be treated as another set of $K$ groups resulting in the following expression that contains a total of $2K$ groups
\begin{align}
  \hat{Q}^{GACV} ~&=~ \sum_{k=1}^K \omega_{0}^{k,e} \hat{Q}_0^k +
  \sum_{k=1}^{K}\sum_{\ell=1}^L  \omega_{\ell}^{k,e}\hat{Q}_{\ell}^{k,e}
  + \sum_{k=K+1}^{2K}\sum_{\ell=1}^L \omega_{\ell}^{k,\mu}\hat{Q}_{\ell}^{k,\mu}
  ~=~ \sum_{k=1}^K \omega_{0}^{k,e} \hat{Q}_0^k + \sum_{k=1}^{2K}\sum_{\ell=1}^L\tilde{\vec{\beta}}_{\ell}^{k}\hat{Q}_{\ell}^{k},
\end{align}
where the last equality combines the two double summations of the previous equality by setting the weights $\tilde{\vec{\beta}}^k_{\ell}=\omega_{\ell}^{k,e}$
if $k \leq K$ and $\omega_{\ell}^{k, \mu}$ otherwise, and by defining the estimator $\hat{Q}_{\ell}^k$ as $\hat{Q}_{\ell}^{k,e}$ if $k \leq K$ and $\hat{Q}_{\ell}^{k,\mu}$ otherwise.  The following definition extends this idea developed for two separate groupings to a more general %grouping
strategy that can be used with $K$ groupings.

\begin{definition}[Generalized Linear Grouped ACV estimator] \label{def:gacv}
  The {\it generalized grouped ACV estimator} is defined by $K$ model groupings consisting of subsets of estimators of $L+1$ models. The $k$-th group is defined by a linear combination of estimators for models defined in $\set{S}^k \in 2^{[L]} \, \backslash \, \emptyset $ of size $|\set{S}^k| = n_k$ according to
\begin{align}
  \hat{Q}^{GACV}(\set{B}) ~=~ \sum_{k=1}^{K} \sum_{\ell \in \set{S}^k}  \vec{\beta}_{\vec{\lambda}^{k^{-1}}(\ell)}^{k}\hat{Q}_{\ell}^{k} ~=~  \sum_{k=1}^{K} \sum_{\ell=0}^L  \tilde{\vec{\beta}}_{\ell}^{k}\hat{Q}_{\ell}^{k},
  \label{eq:GACV}
\end{align}
where $\set{B} = \{\vec{\beta}^1, \ldots, \vec{\beta}^K\}$ with $\vec{\beta}^k \in \reals^{n_k}$. Furthermore, $\vec{\lambda}^k$ and $\mat{R}^{k}$ are the multi-index and restriction operators associated with each $\set{S}^k$, as discussed in Section~\ref{sec:notation}. Finally, $\tilde{\vec{\beta}}^k$ is defined as in Equation~\eqref{eq:beta_tilde}.
\end{definition}

Note this estimator appears identical to~\eqref{eq:blue2}, with this exception that $\hat{Q}_{\ell}^k$ {\it need not be a simple Monte Carlo estimator}. For example, it could be an importance sampling estimator or a more complex estimator derived as some function obtained by first constructing a surrogate model. %% 
A second, and more critical, difference is that the estimators across groups {\it need not be independent}. Indeed each group can even use an identical set of models, but with different estimators. %
Nevertheless, the linear structure yields the same requirements on the weights as the ML-BLUE estimator to ensure unbiasedness. Additionally, the estimator can be written as an ACV with essentially the same argument as in Section~\ref{sec:mlblue_is_acv}. Note that each estimator $\hat{Q}_{\ell}^k$ can itself be potentially optimized; for example, one could use an ensemble estimator that serves to connect ACV and ML-BLUE estimators where the weights of this ensemble can be optimized. Moreover, the sample allocation can be optimized for each estimator and the ensemble allocation can be optimized as well. Prior to analyzing this estimator, we provide some connections to existing approaches.

Note that the GACV is in fact {\it the most general linear ansatz} for any control-variate type variance reduction scheme. If we simply set the estimator of each group to be a particular evaluation of a model $\hat{Q}_{\ell}^k = Q_{\ell}^k(z_{1}^k)$, then we allow any linear combinations of this model. Thus, the GACV can be viewed as simply combining individual evaluations of each of the $L+1$ models -- regardless of whether they are evaluated at the same inputs or not, and regardless of the combinations of models evaluated at a particular input. 

As just one other example available in the literature, consider the Ensemble ACV estimator~\cite{pham2022}.
The ensemble ACV is a variant  of the proposed GACV estimator with a specific choice of weights. This estimator is obtained by first rewriting~\eqref{eq:GACV} as 
\begin{align}
\hat{Q}^{GACV} &= \sum_{k=1}^K\left( \omega_{0}^{k,e} \hat{Q}_0^k + \sum_{\ell=1}^L \left(
\omega_{\ell}^{k,e}\hat{Q}_{\ell}^{k,e} - \omega_{\ell}^{k,\mu}\hat{Q}_{\ell}^{k,\mu}\right)\right)
\label{eq:sum_acv},
\end{align}
and then making the following assumptions: (1) each of the $K$ groups involves the same models, but independent estimators; (2) the high-fidelity estimators are equally weighted across groups, $\omega_{0}^{k,e} = \frac{1}{K}$; and (3) the weights of a given model across all groups is identical, $\omega_{\ell}^{k,e} = \omega_{\ell}^{k, \mu} = \frac{\vec{\alpha}_{\ell}}{K}.$ Under these assumptions we have
\begin{align}
\hat{Q}^{\text{ensemble}-ACV}(\vec{\alpha}) &= \frac{1}{K}\sum_{k=1}^K\left( \hat{Q}_0^k + \sum_{\ell=1}^L \vec{\alpha}_{\ell} \left(
\hat{Q}_{\ell}^{k,e} - \hat{Q}_{\ell}^{k,\mu}\right)\right),
\label{eq:ensemble_acv}
\end{align}
precisely the ensemble ACV estimator given by~\cite[Eq. 34]{pham2022}. We refer to that paper for analysis of the specific estimator, and how the weights $\vec{\alpha}$ can be chosen.

\section{Analysis}\label{sec:analysis}

In this section, we analyze the generalized estimator~\eqref{eq:GACV} by computing its variance, and then defining an optimal selection of weights. First we note that, for the estimator to be unbiased, we again require
\begin{equation}
  \sum_{k=1}^K \tilde{\vec{\beta}}_{0}^{k} = 1 \quad \textrm{ and } \quad \sum_{k=1}^K \tilde{\vec{\beta}}_{\ell}^{k} = 0, \quad \text{ for } \ell=1,\ldots,L
  \label{eq:unbiasedness}
\end{equation}
under the assumption that the estimator for each model $\ell$, $\hat{Q}_{\ell}^k$, has the same bias for each $k$.

Next let us define some covariances between and within groups. First the covariance matrix of all estimators in group $k$ is denoted by
\begin{equation}
  \mat{C}^k = \vari{\vec{\hat{Q}}^{k}} =
  \begin{bmatrix}
    \vari{\hat{Q}_{\vec{\lambda}^k_1}} & \covi{\hat{Q}_{\vec{\lambda}^k_1}, \hat{Q}_{\vec{\lambda}^k_2}} & \cdots & \covi{\hat{Q}_{\vec{\lambda}^k_1}, \hat{Q}_{\vec{\lambda}^k_{n_k}}} \\
    & \ddots &  & \vdots \\
    &  & \ddots &  \\
    \normalfont{\text{Sym}} &  &   & \vari{\hat{Q}_{\vec{\lambda}^k_{n_k}}}
  \end{bmatrix}
  \in \reals^{n_k \times n_k}.
\end{equation}
Next, the covariance between estimators in groups $k$ and $k^{\prime}$ is given by
\begin{equation}
  \mat{C}^{kk^{\prime}} =
  \begin{bmatrix}
    \covi{\hat{Q}_{\lambda_1^k}^{k}, \hat{Q}_{\lambda_{1}^{k^\prime}}^{k^{\prime}}} %% 
    & \cdots & \covi{\hat{Q}_{\lambda_1^k}^{k}, \hat{Q}_{\lambda_{n_{k^{\prime}}}^{k^\prime}}^{k^{\prime}}} \\
    \vdots  & \ddots & \vdots \\
    \covi{\hat{Q}_{\lambda_{n_k}^k}^{k}, \hat{Q}_{\lambda_{1}^{k^\prime}}^{k^{\prime}}} & \cdots & 
    \covi{\hat{Q}_{\lambda_{n_k}^k}^{k}, \hat{Q}_{\lambda_{n_{k^{\prime}}}^{k^\prime}}^{k^{\prime}}}
  \end{bmatrix} \in \reals^{n_k \times n_{k^{\prime}}}.
\end{equation}
Together, these covariance matrices can be written as a single matrix 

  \begin{equation}
    \mat{C} =     \begin{bmatrix}
      \mat{C}^1   & \mat{C}^{12}& \cdots &              & \mat{C}^{1K}\\
                  & \mat{C}^2  & \mat{C}^{23}  & \cdots          & \vdots\\
                  &            & \ddots &              &        \\
                  &            &        & \mat{C}^{K-1} & \mat{C}^{(K-1)K}\\
      \text{Sym}  &            &        &              & \mat{C}^K
    \end{bmatrix}
    \label{eq:mat_C}
  \end{equation}
  Note that, to obtain an optimal set of weights, we will assume that the individual model estimators of each group are structured so that the covariance among them is invertible.

  The variance of the estimator can be written as a function of these matrices. This is summarized by the following result.
  \begin{proposition}[Variance of the GACV]\label{prop:var_gacv}
    The variance of the GACV estimator in Definition~\ref{def:gacv} is given by
    \begin{equation}
      \vari{\hat{Q}^{GACV}(\set{B})} = \vec{\beta}^T \mat{C} \vec{\beta}, \qquad \textrm{ where  }  \qquad    \vec{\beta} =
      \begin{bmatrix}
        \vec{\beta}^1 \\ 
        \vec{\beta}^2 \\
        \vdots \\
        \vec{\beta}^K
      \end{bmatrix},
    \end{equation}
    and $\mat{C}$ is given by~\eqref{eq:mat_C}.
  \end{proposition}
  \begin{proof}
    The proof simply applies the property of covariance of a sum of random variables.
      \begin{align}
    \vari{\hat{Q}^{GACV}(\set{B})} &= \sum_{k=1}^K \sum_{\ell \in \set{S}^k} \sum_{k^{\prime}=1}^K \sum_{\ell^{\prime} \in \set{S}^{k^{\prime}}} \tilde{\vec{\beta}}_{\ell}^{k} \tilde{\vec{\beta}}_{\ell^{\prime}}^{k^{\prime}} \covi{\hat{Q}_{\ell}^k,\hat{Q}_{\ell^{\prime}}^{k^{\prime}}}.
  \end{align}
  Vectorizing this expression yields the stated result.
  \end{proof}

  We now seek to find the optimal weights $\set{B}^*$ that are defined to have minimum variance and zero bias
 \begin{equation}
  \set{B}^* = \argmin_{\set{B}} \vari{\hat{Q}^{GACV}(\set{B})} \quad \text{subject to} \quad  \sum_{k=1}^K \tilde{\vec{\beta}}_{0}^{k} = 1 \quad \textrm{ and } \quad \sum_{k=1}^K \tilde{\vec{\beta}}_{\ell}^{k} = 0, \quad \text{ for } \ell=1,\ldots,L.
  \label{eq:obj_gacv}
\end{equation}
 The solution to this problem is given by the following theorem.

 \begin{theorem}[Optimal weights of the GACV]\label{thm:opt_gacv}
   The solution to~\eqref{eq:obj_gacv} is given by the stacked weight vector
  \begin{equation} \label{eq:opt_beta}
    \vec{\beta}^* =  \mat{C}^{-1} \mat{R}^T  \left(\mat{R} \mat{C}^{-1}\mat{R}^T\right)^{-1} \vec{e}^0,
  \end{equation}
  where $\vec{e}^0 \in \reals^{L+1}$ is the standard basis having the first element equal to one and the rest equal to zero,
  $\mat{C}$ is given in~\eqref{eq:mat_C}, and 
  \begin{equation}
    \mat{R} = \begin{bmatrix} \mat{R}^{1^T} & \cdots & \mat{R}^{k^T} \end{bmatrix} \in \reals^{L+1 \times \prod_{k=1}^K n_{k}}
  \end{equation}
  is a horizontal concatenation of all the restriction matrices.
  Furthermore, the corresponding estimator variance is
  \begin{equation} \label{eq:opt_var}
    \vari{\hat{Q}^{GACV}(\set{B}^*)} =  (\vec{e}^{0})^T  \left(\mat{R} \mat{C}^{-1}\mat{R}^T\right)^{-1} \vec{e}^0
  \end{equation}
\end{theorem}

 \begin{proof}
   Using Proposition~\ref{prop:var_gacv}, the Lagrangian form of~\eqref{eq:obj_gacv} becomes 
  \begin{equation}
    L(\set{B}, \vec{\gamma}) = \vec{\beta}^T \mat{C} \vec{\beta} + 2 \vec{\gamma}^T \left(\mat{R} \vec{\beta} - \vec{e}^0\right).
  \end{equation}
  With setting the Lagrangian derivative to zero,
  \begin{equation}
    \nabla_{\set{B}} L(\set{B}, \vec{\gamma}) = 2 \mat{C}\vec{\beta} + 2 \mat{R}^T \vec{\gamma} = 0,
  \end{equation}
  we obtain $\vec{\beta} = - \mat{C}^{-1} \mat{R}^T \vec{\gamma}.$ Plugging this expression into the Lagrangian, we can then seek to maximize the dual
  \begin{equation}
    g(\vec{\gamma}) =  \vec{\gamma}\mat{R} \mat{C}^{-1} \mat{R}^T \vec{\gamma} - 2 \vec{\gamma}^T \mat{R} \mat{C}^{-1}\mat{R} \vec{\gamma} - 2 \vec{\gamma}^T \vec{e}^0 = -\vec{\gamma}\mat{R} \mat{C}^{-1} \mat{R}^T \vec{\gamma} - 2 \vec{\gamma}^T \vec{e}^0
  \end{equation}
  The maximum occurs at $\vec{\gamma} = -\mat{R} \mat{C}^{-1}\mat{R}^T \vec{e}^0.$ Using this equation in the expression for the weights, we obtain the stated result.
\end{proof}

 We can specialize this result to the case of independent groups by realizing that this change causes $\mat{C}$ to be block diagonal. Plugging this structure into Theorem~\ref{thm:opt_gacv} yields the following corollary.

 \begin{corollary}[Optimal weights under independent groups]
   The optimal weights for a GACV estimator that uses $K$ independent groups of estimators is given by 
  \begin{equation}
    \vec{\beta}^k =  \mat{C}^{k^{-1}} \mat{R}^{k}\left(\sum_{k=1}^K \mat{R}^{k^T} \mat{C}^{k^{-1}} \mat{R}^{k}\right)^{-1} \vec{e}^0,
    \label{eq:gacv_ind_weights}
  \end{equation}
  and the resulting variance is 
  \begin{equation}
    \vari{\hat{Q}^{GACV}(\set{B}^*)} = \vec{e}^{0^T} \left(\sum_{k=1}^K \mat{R}^{k^T} \mat{C}^{k^{-1}} \mat{R}^{k}\right)^{-1}  \vec{e}^0.
    \label{eq:gacv_ind_variance}
  \end{equation}
 \end{corollary}

Note that while these results look similar to the ML-BLUE optimal weights and variance, they are actually a generalized case because they handle arbitrary estimators. The precise ML-BLUE estimator is recovered by assuming Monte Carlo estimators of the expectation. In this case, the covariance of each group of estimators $\mat{C}^k$ is directly related to the covariance of the underlying random variables $\mat{\hat{C}}^k$ in the group through the simple relation $\mat{C}^k = \frac{1}{m^k} \mat{\hat{C}}^k$. Plugging this relation into~\eqref{eq:gacv_ind_weights} and~\eqref{eq:gacv_ind_variance} recovers the corresponding solution to the least-squares problem described in~\cite[Eq. 2.7 and 2.8]{schaden2020}. 
\begin{corollary}[Optimal ML-BLUE estimator]\label{corr:blue_variance}
  Let each estimator within a group be a Monte Carlo estimator for the expectation using shared samples across the group $\hat{Q}_{\ell}^k = \frac{1}{m_k} \sum Q_{\ell}(z_{i}^k).$ Furthermore, let the covariance of all models within a group be $\mat{\hat{C}}^k$ and let each of the $K$ groups use independent samples. Then the optimal GACV estimator is given by
  \begin{equation}
    \vec{\beta}^k =  m^{k} \mat{\hat{C}}^{k^{-1}} \mat{R}^{k}\left(\sum_{k=1}^K m^{k} \mat{R}^{k^T}\mat{\hat{C}}^{k^{-1}} \mat{R}^{k}\right)^{-1} \vec{e}^0,
  \end{equation}
  and the resulting variance is 
  \begin{equation}
    \vari{\hat{Q}^{GACV}(\set{B}^*)} = \vec{e}^{0^T} \left(\sum_{k=1}^K m^k \mat{R}^{k^T} \mat{\hat{C}}^{k^{-1}} \mat{R}^{k}\right)^{-1}  \vec{e}^0.
  \end{equation}
\end{corollary}

\section{Numerical experiments}~\label{sec:benefits}

In this section, we turn to an empirical investigation of the benefits of non-independent group ACV estimators, as compared to the special case of an independent group ML-BLUE estimator. Our initial aim is to simply demonstrate {\it existence} of cases where non-independent groupings are beneficial. The ramifications of this investigation suggest that (1) significant open questions remain regarding how best to distribute resources in a linear multifidelity variance reduction scheme; (2) the concept of a BLUE %``best linear unbiased estimator''
must be understood to refer to the best of a certain sub-class of linear models, not a general statement pertaining to all possible estimators formed by linear combinations of model evaluations.

We begin this investigation in Section~\ref{sec:acv_is_acv_mf} by comparing two classical ACV estimators, ACV-IS and ACV-MF, the former of which is an ML-BLUE.
%, which use the same model groups but have different independence properties.
Specifically, we show that there exists combinations of correlations and model costs for which removing the assumption of grouping independence required by ML BLUE results in smaller estimator variance. Next, in Section~\ref{sec:general}, we demonstrate how model groupings of an ML-BLUE can be converted into a non-independent grouping estimator, and that again there exists settings where the non-independent GACV has better variance reduction.

\subsection{ACV-IS vs. ACV-MF}~\label{sec:acv_is_acv_mf}

As an initial example, we demonstrate the benefit of moving beyond the independent groupings used by ML-BLUE and employing groups with shared (non-independent) samples by comparing the ACV-IS and ACV-MF estimators.
As shown in~\cite{schaden2020}, for cases with more than two models, the ACV-IS estimator is a BLUE %%
while the ACV-MF estimator is not. Both of these estimators are special cases of the GACV.

The ACV-IS discussed in~\cite{schaden2020} consists of the groups $\set{S}= \{ [L] \} \cup_{i=1}^L \{\{i\}\},$ and takes the form
\begin{equation}
  \hat{Q}^{ACV-IS} = \frac{1}{m^1} \sum_{i=1}^{m^1} Q(z^1_i) + \sum_{\ell=1}^L \alpha_{\ell} \left( \frac{1}{m^1} \sum_{i=1}^{m^1} Q_{\ell}(z^1_i) - \frac{1}{m^{\ell}} \sum_{i=1}^{m^{\ell}} Q_{\ell}(z^{\ell}_i)\right), \quad \text{ where } \set{Z}^i \bigcap \set{Z}^j = \emptyset \text{ when } i \neq j.
\end{equation}
ACV-MF~\cite{gorodetsky2020} consists of the same groups, but the second term of the control variate re-uses the input-samples that were used by the higher fidelity models:
\begin{equation}
  \hat{Q}^{ACV-MF} = \frac{1}{m^1} \sum_{i=1}^N Q(z^1_i) + \sum_{\ell=1}^L \alpha_{\ell} \left( \frac{1}{m^1} \sum_{i=1}^N Q_{\ell}(z^1_i) - \frac{1}{{m}^{\ell}} \sum_{i=1}^{{m}^{\ell}} Q_{\ell}(z^{\ell}_i)\right), \quad \text{ where } \set{Z}^i \subset \set{Z}^j \text{ when } i < j.
\end{equation}
Because the samples are reused, the estimators in each group are not independent and therefore ACV-MF is not a BLUE based on the definition in~\cite{schaden2020}. %
Note also that the number of samples in the groups for ACV-MF is higher than in ACV-IS, even though the overall cost is the same.

We perform several experiments on the three model setting where both the ACV-IS and ACV-MF sets are
\[\set{S}=\{\{{0, 1, 2\}, \{1\}, \{2}\}\}.\]
In this setting, we consider an equal cost allocation across the two estimators. Thus, we let $n$ denote the total number of high-fidelity evaluations, $m_1$ denote the total number of evaluations of the first low-fidelity model $Q_1$, and $m_2$ denote the total number of evaluations of the second low-fidelity model $Q_2$. This setup implies the following number of samples per group: for ACV-IS, group 1 has $m^1 = n$, group two has a unique set $m^2 = m_1 - n$, and group three has a unique set $m^3 = m_2 - n$; and for ACV-MF, group 1 has $m^1 = n$, group two has all available mid-fidelity samples $m^2 = m_1$, and group three has all available low-fidelity samples $m^3 = m_2$.

We compare the variance of these estimators in Figure~\ref{fig:acv_is_vs_acv_mf} for $n=5$, a range of $m_1$ and $m_2$, and two different correlation settings.\footnote{The variance of each model does not affect the ratio of estimator variances.} The left panel corresponds to a correlation structure where the correlation between $Q$ and $Q_2$ is lower than the correlation between $Q_1$ and $Q_3$, whereas the right panel has the opposite structure. The qualitative behavior between the two cases is different. In the left panel, the benefit of ACV-MF over ACV-IS is maximized for smaller numbers of $Q_1$ samples and larger numbers of $Q_2$ samples. The red lines in this plot indicate unity contours which correspond to transitions between estimator preference. Thus there exist a large regime where ACV-MF has lower variance than ACV-IS for this case. The right panel indicates that this behavior is highly dependent on the correlation structure. Here, ACV-IS has lower variance over all sample allocations, but its benefit decreases as the number of samples of $Q_1$ increases. %%

Overall, this example demonstrates that just because an estimator is a BLUE does not mean that it makes the best use of all available resources. Clearly, there exist cases where reusing samples across groups yields better variance reduction for a given budget. In the next section, we show more generally how a GACV estimator that reuses samples can be constructed automatically from certain common group and sample allocation schemes used in ML-BLUE, and that this construction can yield variance reduction.

\begin{figure}
  \centering
  \begin{subfigure}{0.45\textwidth}
    \centering
    \includegraphics[scale=1]{{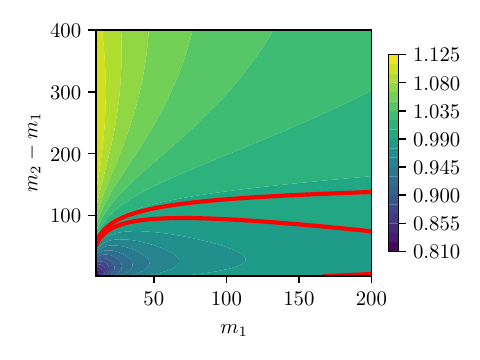}}
    \caption{$\rho_{01} = 0.95, \rho_{02} = 0.8, \rho_{12} = 0.9.$}
  \end{subfigure}
  ~
  \begin{subfigure}{0.45\textwidth}
    \centering
    \includegraphics[scale=1]{{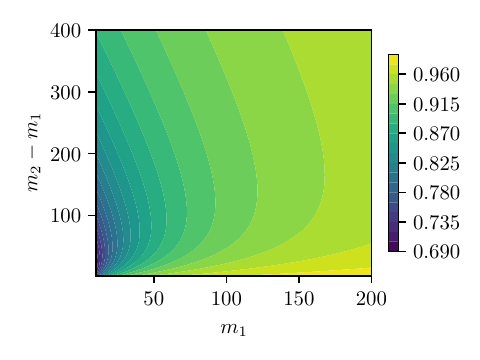}}
    \caption{$\rho_{01} = 0.95, \rho_{02} = 0.93, \rho_{12} = 0.9.$}
  \end{subfigure}
  \caption{Ratio of estimator variance $\vari{\hat{Q}^{ACV-IS}} / \vari{\hat{Q}^{ACV-MF}}$ for fixed cost using a three model setting for two choices of correlation structures as a function of number of evaluations of each model. The $y$ axis is the number of additional samples used for $Q_2$ over $Q_1$. %% 
     Red lines indicate a variance ratio of one, values higher than one indicate the ACV-MF estimator obtains greater variance reduction. These results indicate that ACV-MF estimator can out-perform the ACV-IS estimator (which is BLUE), depending on correlation and computational cost of models.}
  \label{fig:acv_is_vs_acv_mf}
\end{figure}

\subsection{Converting an arbitrary ML-BLUE estimator into a nested GACV estimator}\label{sec:general}

In this section we first provide an algorithm that converts certain ML-BLUE sample allocations into GACV sample allocations. We then show numerical experiments that indicate that this GACV estimator can outperform an optimal ML-BLUE estimator.

\subsubsection{Nested sample grouped estimator}

In situations with well-correlated approximate models, the optimal groups for the ML-BLUE estimator often follow a recursive structure. For example, in the case where all combination of models in a group are allowed, the optimal structure often becomes fully nested so that $\set{S}^1 = [L]$, and $\set{S}^{k} = \{k-1,\ldots,L\}$ for $k = 2,\ldots, L+1$ ~\cite{schaden2020,schaden2021}. Because it may be computationally expensive to search over the space of full models, the SAOB-M estimator was introduced in the same work to limit each group to at most $1 \leq M \leq L+1$ models. In this case, there still exists a specific type of nested structure $\set{S}^1 = [M-1]$, and $\set{S}^{k} = \{k-1,\ldots,\min(k+M-2, L)\}$ for $k=2,\ldots,L+1$ --- see e.g.,~\cite[Fig. 6]{schaden2021}. Note that SAOB-L would correspond to the fully nested estimator. A second viewpoint of the same structure is that the $\ell$-th model belongs to groups $\set{S}^{\max(\ell-M,1)}, \ldots, \set{S}^{\ell+1}.$ As an example, for a SAOB-3 estimator for $L=4$, the groups are
\begin{equation}
  \set{S}^1 = \{0,1,2\}, \quad \set{S}^2 = \{1,2,3\}, \quad \set{S}^3 = \{2,3,4\}, \quad \set{S}^4 = \{3,4\}, \quad \set{S}^5 = \{4\}.
\end{equation}

In this section, we provide a non-independent GACV estimator that has these identical groupings, but re-distributes the sample allocation amongst non-independent groups to achieve an estimator of equivalent cost. Consider a (potentially optimal) sample allocation scheme for an ML-BLUE estimator that suggests $m^k$ samples for group $\set{S}^k.$ Our nested estimator will determine new sample allocations $\hat{m}^k$ that respects a certain nested sample structure. Specifically, if $\hat{Z}^k$ are the samples used for group $k$, then we construct GACV estimators that use $\set{Z}^{k+1} \subset \set{Z}^k$. Moreover, each of these groupings will now have $\hat{m}^k = \left| \set{Z}^k \right|$ samples.

To enable a valid comparison, we must ensure that the total number of evaluations of each model is consistent between the ML-BLUE estimator and the GACV estimator with equivalent model groups but a nested sample structure that induces dependencies. In other words, we construct a mapping $\hat{m}^k = f_k(m^1,\ldots,m^k)$ such that the total number of evaluations of each model are equal. With this goal, let $n^\ell$ define the total number of samples of model $\ell$ such that for a ML-BLUE allocation
\begin{equation}
  n^\ell = \sum_{k=1, \  \ell \in \set{S}^k }^{L+1} m^k = \sum_{k=\max(\ell-M+2,1)}^{\ell+1} m^k
  \label{eq:ml_blue_alloc}
\end{equation}
and for the associated GACV allocation
\begin{equation}
  n^{\ell} = \max\{\hat{m}^k: \ \text{s.t.} \ \ell \in \set{S}^k  \} = \hat{m}^{\ell+1},
  \label{eq:gacv_alloc}
\end{equation}
for $\ell = 0,\ldots,L.$ The last equality for the GACV allocation exploits the nested property of the input sets for that estimator. We begin with $\hat{m}^1 = m^1.$ Then
Equating~\eqref{eq:ml_blue_alloc} with~\eqref{eq:gacv_alloc} we obtain
\begin{equation}
  \hat{m}^{\ell+1}= \sum_{k=\max(\ell-M+2,1)}^{\ell+1} m^k, \quad \ell = 1, \ldots, L.
  \label{eq:conversion}
\end{equation}
Table~\ref{tab:example_alloc} shows a concrete example for a prototypical example of a 5 model case with a maximum grouping of three models $(M=3)$.

\begin{table}
  \centering
  \caption{Example SAOB-3 allocation for a 5 model case ($L=4$). The ML-BLUE allocation is [5,5,5,7,18] for groups $\set{S}^1$ through $\set{S}^5$. The number of evaluations of each model in each group is shown for both the ML-blue allocation and the equivalent GACV allocation using~\eqref{eq:conversion}. Each element shows the ML-BLUE (MLB) allocation and the GACV allocation (GACV) for each model in each group. The bottom row shows that the sum of the number evaluations is equal using~\eqref{eq:ml_blue_alloc} and~\eqref{eq:gacv_alloc}.}~\label{tab:example_alloc}
  \begin{tabular}{|l|c|c|c|c|c|c|c|c|c|c|}
    \hline 
    \backslashbox{Group}{Model}  & \multicolumn{2}{c|}{0}  &  \multicolumn{2}{c|}{1}   &  \multicolumn{2}{c|}{2}    &  \multicolumn{2}{c|}{3}   &  \multicolumn{2}{c|}{4}     \\
    \hline
                        & MLB & GACV & MLB & GACV & MLB & GACV & MLB & GACV & MLB & GACV  \\
    \hline
    \hline
    $\set{S}^1$  & 5 & 5   & 5 & 5    & 5 & 5   & 0 & 0   & 0 & 0     \\
    $\set{S}^2$  & 0 & 0   & 5 & 10  & 5 & 10 & 5 & 10 & 0 & 0     \\
    $\set{S}^3$  & 0 & 0   & 0 & 0    & 5 & 15 & 5 &15  & 5 & 15   \\
    $\set{S}^4$  & 0 & 0   & 0 & 0    & 0 & 0   & 7 & 17 & 7 & 17   \\
    $\set{S}^5$  & 0 & 0  & 0  & 0    & 0 & 0   & 0  & 0  & 18 & 30  \\
    \hline
    \hline
    Number of evals $n^\ell$      &  5 & 5  & 10 & 10  & 15 & 15   & 17 & 17  & 30 & 30 \\
    \hline
  \end{tabular}
\end{table}
To summarize, an allocation for a SAOB-M grouping can be converted into a concrete GACV nested sample estimator according to Algorithm~\ref{alg:nested_gacv}. Thus, this algorithm can serve as a drop-in replacement for ML-BLUE estimators with this structure. %%
%
%[Algorithm 1]
\begin{algorithm}
	\caption{Nested-sample GACV Monte Carlo estimator conversion from SAOB-M estimator}
	\label{alg:nested_gacv}
	\begin{algorithmic}[1]
		\Input
        \Desc{$M$}{$\quad$ Maximum size of model group in SAOB-M estimator} 
		\Desc{$m^k$ for $k=1,\ldots,L+1$}{$\quad$ Sample allocations for each SAOB-Mgroup} 
		\EndInput
		\Output 
		\Desc{$\hat{Q}_{\ell}^k$ for $\ell=0,\ldots,L$ and $k=1,\ldots,L+1$}{$\quad$ Nested-sample estimators within the GACV~\eqref{eq:GACV}}
        \Desc{$\hat{m}^k$ for $k=1,\ldots,L+1$}{$\quad$ Number of samples in equivalent cost nested GACV estimator} 
		\EndOutput
        \State \textcolor{grey}{\textbf{Compute the number of evaluations of each group}}
        \State $\hat{m}^1 = m^1$
		\For{$\ell=1$ to $L$}
		\State $\hat{m}^{\ell+1} \leftarrow $ Equation~\eqref{eq:conversion}
		\EndFor
		\State $\vec{z} = \{z_1,\ldots, z_{\hat{m}^{L+1}}\}$ \Comment{Generate all inputs to be used}
		\State \textcolor{grey}{\textbf{Compute the estimators in each group}}
		\For {$\ell=0$ to $L$}
		\For {$k = 1$ to $L+1$}
		\If {$\ell \in \set{S}^k$}
		\State $\hat{Q}_{\ell}^k = \frac{1}{\hat{m}^k} \sum_{a=1}^{\hat{m}^{k}} Q_{\ell}(z_{a})$
		\EndIf
		\EndFor
		\EndFor	
	\end{algorithmic}
\end{algorithm}
The variance of this estimator is given by the following corollary.
\begin{corollary}[Nested-sample GACV estimator from independent-sample ML BLUE]\label{corr:nested_gacv_samples}
Let $\mat{\hat{C}}^k$, $k=1,\ldots,K$ denote the covariance between the models within group $k$ and $\mat{\hat{C}}^{kk^{\prime}}$ denote the covariance of models between models of group $k$ and $k^{\prime}$.  Then the optimal GACV estimator computed using Algorithm~\ref{alg:nested_gacv} is given by optimal weights~\eqref{eq:opt_beta} with corresponding variance~\eqref{eq:opt_var}, where the covariance matrix is
 \begin{equation}
    \mat{C} =     \begin{bmatrix}
      \frac{1}{\hat{m}^1}\mat{\hat{C}}^1   &  \frac{1}{\max(\hat{m}^1, \hat{m}^{2})}  \mat{\hat{C}}^{12}& \cdots &              &  \frac{1}{\max(\hat{m}^1, \hat{m}^{K})}\mat{\hat{C}}^{1K}\\
                  & \frac{1}{\hat{m}^2} \mat{\hat{C}}^2  &  \frac{1}{\max(\hat{m}^2, \hat{m}^{3})} \mat{\hat{C}}^{23}  & \cdots          & \vdots\\
                  &            & \ddots &              &        \\
                  &            &        & \frac{1}{\hat{m}^{K-1}}\max{\hat{C}}^{K-1} &  \frac{1}{\max(\hat{m}^{K-1}, \hat{m}^{K})} \mat{\hat{C}}^{(K-1)K}\\
      \text{Sym}  &            &        &              & \frac{1}{\hat{m}^K}\mat{\hat{C}}^K
    \end{bmatrix}
\end{equation}	
\end{corollary}
\begin{proof}
The Monte Carlo sampling strategy implies that the covariance of the estimators are directly related to the covariance of the underlying models. Specifically,  The covariances in $\mat{C}^{kk^{\prime}}$ become
\begin{equation}
  \covi{\hat{Q}_{\lambda_i^k}^{k}, \hat{Q}_{\lambda_{j}^{k^\prime}}^{k^{\prime}}} = \covi{\frac{1}{\hat{m}^k} \sum_{a=1}^{\hat{m}^k}Q_{\lambda_{i}^k}( z_{a} ), \frac{1}{\hat{m}^{k^{\prime}}}  \sum_{a=1}^{\hat{m}^{k^{\prime}}} Q_{\lambda_{j}^{k^{\prime}}}( z_{a} )}
\end{equation}
Assume without loss of generality that $\hat{m}^k \leq \hat{m}^{k^{\prime}}$. Then we have
\begin{align}
  \covi{\hat{Q}_{\lambda_i^k}^{k}, \hat{Q}_{\lambda_{j}^{k^\prime}}^{k^{\prime}}} &= \covi{\frac{1}{\hat{m}^k} \sum_{a=1}^{\hat{m}^k}Q_{\lambda_{i}^k}( z_{a} ), \frac{1}{\hat{m}^{k^{\prime}}}  \sum_{a=1}^{\hat{m}^{k}} Q_{\lambda_{j}^{k^{\prime}}}( z_{a} )
  + \frac{1}{\hat{m}^{k^{\prime}}}  \sum_{a=\hat{m}^{k}+1}^{\hat{m}^{k^{\prime}}} Q_{\lambda_{j}^{k^{\prime}}}( z_{a})} \\
  &= \frac{1}{\hat{m}^k \hat{m}^{k^{\prime}}} \covi{ \sum_{a=1}^{\hat{m}^k}Q_{\lambda_{i}^k}( z_{a} ), \sum_{a=1}^{\hat{m}^{k}} Q_{\lambda_{j}^{k^{\prime} } }}
 = \frac{1}{\max(\hat{m}^k, \hat{m}^{k^{\prime}})} \covi{Q_{\lambda_i^k}, Q_{\lambda_{j}^k}},
\end{align}
while the covariances in $\mat{C}^k$ are the same as those in Corollary~\ref{corr:blue_variance}. In this case we would have  $\mat{C}^{kk^{\prime}} = \frac{1}{\max(\hat{m}^k, \hat{m}^{k^{\prime}})} \mat{\hat{C}}^{kk^{\prime}}.$ Combining these results yields the stated covariance. 
\end{proof}

\subsubsection{Numerical demonstration}

We now provide an empirical comparison between a SAOB-M estimator and the derived nested-sample GACV estimator. Because the relative performance of the ML BLUE and nested group ACV estimators depends on the covariance between models and the relative computational cost of each model, we compared the variance of these two estimators for 1000 randomly generated covariance matrices $\mat{\hat{C}}$ and cost vectors $\vec{w}$. The covariance matrices were generating using the \texttt{scipy.stats}~\cite{SciPy-NMeth_2020} function \texttt{random\_correlations} with arguments that generated correlations roughly between 0.8 and 1. Moreover, we fixed the cost of the high-fidelity model to 1, and randomly drew the computational costs of each low fidelity model I.I.D. from the log uniform distribution over the interval $[0.01, 1]$. The covariance matrices were ordered such that the covariance between the high-fidelity model the low-fidelity models decreases with the column index of the matrix. This ordering was enforced to reflect the hierarchical nature of the model groups we considered. The randomly generated costs were also sorted from highest to lowest for the same reason.

For each model setting (covariance and cost) sampled above, we first perform an optimal sample allocation for ML-BLUE, ensure that it assigns samples to each of the model groups, and then compute a derived nested-sample GACV estimator. After this procedure, we compare the theoretical estimator variance for these allocations.

To be specific, the optimal ML-BLUE sample allocation is obtained by solving a semi-definite-program (SDP)~\cite{Croci_WW_CMAME_2023}
\begin{align}\label{eq:mlblue-sdp}
  t^*, \vec{m}^* = \arg \min_{t, \vec{m}\ge 0} \quad \text{s.t} \quad \begin{bmatrix}\sum_{k=1}^K m^k \mat{R}^{k^T} \mat{\hat{C}}^{k^{-1}} \mat{R}^{k} & \vec{e}^0 \\ \vec{e}^{0^T} & t \end{bmatrix} \succeq 0, \qquad \vec{m}^T\vec{w} \le W, \qquad \vec{m}^T\vec{h} \ge 1,
\end{align}
where $\vec{m} = [m^1,\ldots,m^{K}]$ is the sample allocation, $W$ is the total allowable computational cost, $\vec{w}=[w_0,w_1,\ldots,w_L]\in\reals^{L+1}$ is the cost of each model; and $\vec{h}\in\reals^K$ is a vector with $v_k=1$ if the high-fidelity model is in $\set{S}^k$, and $v_k=0$ otherwise.

The following numerical results were generated using the \textrm{PyApprox} package~\cite{PyApprox}.
Figure~\ref{fig:nested_comparison} displays a histogram of the variance ratio of the ML-BLUE SAOB estimator compared to the derived nested GACV estimator for various $(L,M)$ combinations. Note that $(L,L+1)$ would refer to a fully recursive grouping. The main insight from this graph is that there exists a set of problem conditions under which non-independent groupings with sample reuse perform better than independent groupings without sample reuse. Under the experimental conditions explored here, this occurs a majority of the time, with seemingly greater benefit of GACV as $M$ is lowered compared to $L$. Not visible on this plot, as it is hidden behind the line at unity, is the fully recursive case with $(L,M) =(4,5)$ which yields identical performance between the ML-BLUE and nested GACV estimators. Further theoretical investigation is needed to identify what properties of the covariance matrix in Corollary~\ref{corr:nested_gacv_samples} lead to this behavior.

Finally, we see that there are also cases where the ratio of performance between GACV and ML-BLUE reaches a factor of two or more, suggesting that it is worthwhile to check the variance of both estimators (based on a pilot sample or other covariance estimate) to ensure that the best estimator is used in practical applications.

\begin{figure}
  \centering
  \includegraphics[scale=1]{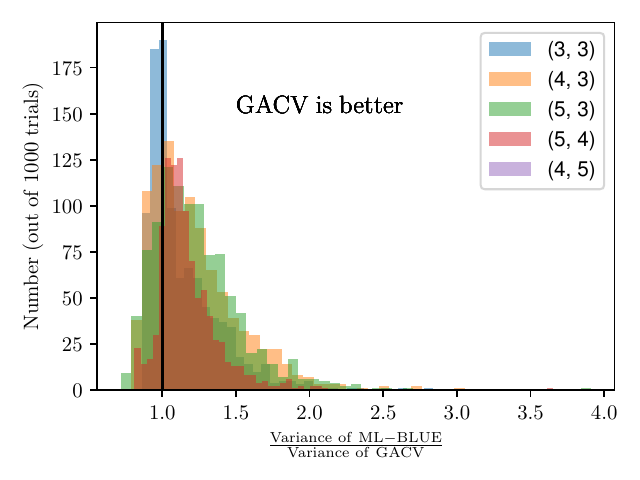}
  \caption{Histogram of the ratio of the ML-BLUE SAOB-M estimator variance to the GACV estimator variance for several $(L,M)$ pairs. The black line denotes equal performance between estimator where experiments to the right of the line indicate that the variance reduction of GACV is better. The GACV often exhibits better variance reduction, and this improvement is typically larger when $M$ is lower compared to $L$. The fully recursive case $(4,5)$ results in a variance ratio of 1.}
  \label{fig:nested_comparison}
\end{figure}

\section{Conclusion}
In this paper, we have extended the approximate control variate estimator to allow for model groupings. We have shown that this estimator encompasses all known ML-BLUE estimators, and is able to represent {\it any} linear combination of model samples. Using this relationship, we have derived a nested-sample GACV estimator from any SAOB-M ML-BLUE estimator and empirically demonstrated that having non-independent sample groupings can yield improvements in variance reduction.

The discovery of the proposed estimator has opened a number of questions. A number of areas for future work are clearly evident to fully realize the potential of this class of estimators.  Future work should address the issue of optimal sample allocation and estimator design to avoid having to perform model selection for choosing %an
a hierarchical estimator based on a pilot set of covariances.
Moreover, new optimization approaches should be developed to enable less restricted sampling schemes than for models with independent groupings.
New schemes should be devised to simultaneously allocation samples and discover optimal groupings.
Finally, extensions to multiple outputs and additional statistics~\cite{dixon2023covariance} should be made for the GACV estimator.

\section{Acknowledgements}
      Alex A. Gorodetsky was supported in part by Sandia National Laboratories and, in part by an NSF CAREER Award CMMI-2238913. John D. Jakeman was supported by the US Department of Energy’s Office of Advanced Scientific Computing Research program.  Michael S. Eldred was supported by the National Nuclear Security Administration's Accelerated Strategic Computing program.

      Sandia National Laboratories is a multi-mission laboratory managed and operated by National Technology \& Engineering Solutions of Sandia, LLC (NTESS), a wholly owned subsidiary of Honeywell International Inc., for the U.S. Department of Energy’s National Nuclear Security Administration (DOE/NNSA) under contract DE-NA0003525. This written work is authored by an employee of NTESS. The employee, not NTESS, owns the right, title and interest in and to the written work and is responsible for its contents. Any subjective views or opinions that might be expressed in the written work do not necessarily represent the views of the U.S. Government. The publisher acknowledges that the U.S. Government retains a non-exclusive, paid-up, irrevocable, world-wide license to publish or reproduce the published form of this written work or allow others to do so, for U.S. Government purposes. The DOE will provide public access to results of federally sponsored research in accordance with the DOE Public Access Plan.

\bibliographystyle{plain}
\bibliography{refs}

\end{document}

%% file: preamble.tex
\usepackage[margin=1in]{geometry}
\usepackage{amsmath}
\usepackage{amssymb}
\usepackage{hyperref}
\usepackage{subcaption}
\usepackage{tikz}
\usepackage{pgfplots}
\usepackage{tcolorbox}
\usepackage{bm}
\usepackage[textsize=tiny]{todonotes}
\usepackage{slashbox}
\usepackage{amsthm}
\usepackage{multirow}
\newtheorem{problem}{Problem}
\newtheorem{theorem}{Theorem}
\newtheorem{definition}[theorem]{Definition}

\newtheorem{proposition}[theorem]{Proposition}

\newtheorem{corollary}[theorem]{Corollary}

\usepackage{algorithm,algpseudocode}
\algblock{Input}{EndInput}
\algnotext{EndInput}
\algblock{Output}{EndOutput}
\algnotext{EndOutput}
\newcommand{\Desc}[2]{\State \makebox[17em][l]{#1}#2}

\DeclareMathOperator*{\argmin}{arg\,min}

\definecolor{grey}{rgb}{0.5, 0.5, 0.49}

%% file: macros.tex
\newcommand{\reals}{\mathbb{R}}

\newcommand{\set}[1]{\mathcal{#1}}
\newcommand{\mat}[1]{\mathbf{#1}}
\renewcommand{\vec}[1]{\bm{#1}}

\newcommand{\vari}[1]{\mathbb{V}\left[#1\right]}
\newcommand{\covi}[1]{\mathbb{C}\text{ov}\left[#1\right]}

%% file: main.bbl
\begin{thebibliography}{10}

\bibitem{Bomarito2022}
Geoffrey~F Bomarito, Patrick~E Leser, James~E Warner, and William~P Leser.
\newblock On the optimization of approximate control variates with
  parametrically defined estimators.
\newblock {\em Journal of Computational Physics}, 451:110882, 2022.

\bibitem{Croci_WW_CMAME_2023}
M.~Croci, K.E. Willcox, and S.J. Wright.
\newblock Multi-output multilevel best linear unbiased estimators via
  semidefinite programming.
\newblock {\em Computer Methods in Applied Mechanics and Engineering},
  413:116130, August 2023.

\bibitem{destouches2023multivariate}
Mayeul Destouches, Paul Mycek, and Selime G{\"u}rol.
\newblock Multivariate extensions of the multilevel best linear unbiased
  estimator for ensemble-variational data assimilation.
\newblock {\em arXiv preprint arXiv:2306.07017}, 2023.

\bibitem{dixon2023covariance}
Thomas~O Dixon, James~E Warner, Geoffrey~F Bomarito, and Alex~A Gorodetsky.
\newblock Covariance expressions for multi-fidelity sampling with multi-output,
  multi-statistic estimators: Application to approximate control variates.
\newblock {\em arXiv preprint arXiv:2310.00125}, 2023.

\bibitem{geraci2017multifidelity}
Gianluca Geraci, Michael~S Eldred, and Gianluca Iaccarino.
\newblock A multifidelity multilevel monte carlo method for uncertainty
  propagation in aerospace applications.
\newblock In {\em 19th AIAA non-deterministic approaches conference}, page
  1951, 2017.

\bibitem{giles2008}
Michael~B Giles.
\newblock Multilevel monte carlo path simulation.
\newblock {\em Operations research}, 56(3):607--617, 2008.

\bibitem{gorodetsky2020}
Alex~A Gorodetsky, Gianluca Geraci, Michael~S Eldred, and John~D Jakeman.
\newblock A generalized approximate control variate framework for multifidelity
  uncertainty quantification.
\newblock {\em Journal of Computational Physics}, 408:109257, 2020.

\bibitem{haji2016}
Abdul-Lateef Haji-Ali, Fabio Nobile, and Ra{\'u}l Tempone.
\newblock Multi-index monte carlo: when sparsity meets sampling.
\newblock {\em Numerische Mathematik}, 132:767--806, 2016.

\bibitem{PyApprox}
J.D. Jakeman.
\newblock Pyapprox: A software package for sensitivity analysis, bayesian
  inference, optimal experimental design, and multi-fidelity uncertainty
  quantification and surrogate modeling.
\newblock {\em Environmental Modelling \& Software}, 170:105825, 2023.

\bibitem{lavenberg1982statistical}
Stephen~S Lavenberg, Thomas~L Moeller, and Peter~D Welch.
\newblock Statistical results on control variables with application to queueing
  network simulation.
\newblock {\em Operations Research}, 30(1):182--202, 1982.

\bibitem{lavenberg1981perspective}
Stephen~S Lavenberg and Peter~D Welch.
\newblock A perspective on the use of control variables to increase the
  efficiency of monte carlo simulations.
\newblock {\em Management Science}, 27(3):322--335, 1981.

\bibitem{peherstorfer2016}
Benjamin Peherstorfer, Karen Willcox, and Max Gunzburger.
\newblock Optimal model management for multifidelity monte carlo estimation.
\newblock {\em SIAM Journal on Scientific Computing}, 38(5):A3163--A3194, 2016.

\bibitem{pham2022}
Trung Pham and Alex~A Gorodetsky.
\newblock Ensemble approximate control variate estimators: Applications to
  multifidelity importance sampling.
\newblock {\em SIAM/ASA Journal on Uncertainty Quantification},
  10(3):1250--1292, 2022.

\bibitem{rubinstein1985efficiency}
Reuven~Y Rubinstein and Ruth Marcus.
\newblock Efficiency of multivariate control variates in monte carlo
  simulation.
\newblock {\em Operations Research}, 33(3):661--677, 1985.

\bibitem{schaden2020}
Daniel Schaden and Elisabeth Ullmann.
\newblock On multilevel best linear unbiased estimators.
\newblock {\em SIAM/ASA Journal on Uncertainty Quantification}, 8(2):601--635,
  2020.

\bibitem{schaden2021}
Daniel Schaden and Elisabeth Ullmann.
\newblock Asymptotic analysis of multilevel best linear unbiased estimators.
\newblock {\em SIAM/ASA Journal on Uncertainty Quantification}, 9(3):953--978,
  2021.

\bibitem{SciPy-NMeth_2020}
Pauli Virtanen, Ralf Gommers, Travis~E. Oliphant, Matt Haberland, Tyler Reddy,
  David Cournapeau, Evgeni Burovski, Pearu Peterson, Warren Weckesser, Jonathan
  Bright, St{\'e}fan~J. {van der Walt}, Matthew Brett, Joshua Wilson, K.~Jarrod
  Millman, Nikolay Mayorov, Andrew R.~J. Nelson, Eric Jones, Robert Kern, Eric
  Larson, C~J Carey, {\.I}lhan Polat, Yu~Feng, Eric~W. Moore, Jake
  {VanderPlas}, Denis Laxalde, Josef Perktold, Robert Cimrman, Ian Henriksen,
  E.~A. Quintero, Charles~R. Harris, Anne~M. Archibald, Ant{\^o}nio~H. Ribeiro,
  Fabian Pedregosa, Paul {van Mulbregt}, and {SciPy 1.0 Contributors}.
\newblock {{SciPy} 1.0: Fundamental Algorithms for Scientific Computing in
  Python}.
\newblock {\em Nature Methods}, 17:261--272, 2020.

\bibitem{yang2023control}
Hang Yang, Yuji Fujii, Kon-Well Wang, and Alex~A Gorodetsky.
\newblock Control variate polynomial chaos: Optimal fusion of sampling and
  surrogates for multifidelity uncertainty quantification.
\newblock {\em International Journal for Uncertainty Quantification}, 13(3),
  2023.

\end{thebibliography}
